%% file: efficientrandomness.tex
\documentclass[english,11pt]{article}

\usepackage[sumlimits]{amsmath}
\usepackage{amsfonts,amssymb,dsfont,amsthm}
\usepackage{times}
\usepackage{enumitem}
\usepackage{type1cm,graphicx}
\usepackage[usenames]{color}
\usepackage{bbm,mathtools}
\usepackage[left=2cm,top=2cm,right=2cm]{geometry}
\usepackage{hyperref}
\newcommand{\tnote}[1]{{\color{red} \footnotesize(Thomas: #1)}}

\input{def}

\newcommand{\eps}{\varepsilon}

\begin{document}

\author{ Urmila Mahadev\thanks{California Institute of Technology, USA. Email: \texttt{mahadev@caltech.edu}} \and Umesh Vazirani\thanks{UC Berkeley, USA. Email: \texttt{vazirani@cs.berkeley.edu}} \and Thomas Vidick\thanks{California Institute of Technology, USA. Email: \texttt{vidick@caltech.edu}}}

\title{Efficient Certifiable Randomness from a Single Quantum Device}

\maketitle

\begin{abstract}
Brakerski et.\ al~\cite{oneproverrandomness} introduced the model of cryptographic testing of a single untrusted quantum device and gave a protocol 
for certifiable randomness generation. We use the leakage resilience properties of the Learning With Errors problem to address a key issue left open in previous work --- the rate of generation of randomness. Our new protocol can certify $\Omega(n)$ fresh bits of randomness in constant  
rounds, where $n$ is a parameter of the protocol and the total communication is $O(n)$, thus achieving a nearly optimal rate. The proof that the output is statistically random is conceptually simple and technically elementary.
\end{abstract}

 \section{Introduction}
 
 The testing of quantum devices poses fundamental difficulties due to their computational power and 
 restrictions imposed by the laws of quantum mechanics on access to their internal state. \cite{oneproverrandomness} introduced a model in which a purely classical verifier uses 
 cryptography to enhance its interactions with a single untrusted and polynomial-time bounded quantum computer. 
 In this model, they gave a protocol for generating certifiably random $n$-bit strings through $\Theta(n)$ iterations
 of a basic ``qubit certification test.'' 
 The correctness of the qubit certification test is based on a cryptographic primitive called a Noisy Trapdoor Claw-free functions (NTCFs) 
 with an adaptive hardcore bit. This primitive in turn was implemented based on the (post-quantum) hardness of the Learning with Errors problem (LWE). The qubit certification test has since been used for other tasks including quantum homomorphic encryption~\cite{mahadev2020classical}, quantum money~\cite{radian2019semi}, and delegated quantum computation~\cite{oneprover}, and has led to many follow-up works and extensions~\cite{chia2020classical,alagic2020non,brakerski2020simpler,hirahara2021test,gheorghiu2022quantum}. 
 
Recall that the qubit certification test from~\cite{oneproverrandomness} allows a classical verifier to reach into a quantum computer's Hilbert space and enforce a qubit structure in it, by ensuring that the untrusted quantum computer creates a $\ket{+}$ state and measures it in the computational basis, thereby generating $1$ bit of randomness. While the analysis of \cite{oneproverrandomness} is powerful, in that it uses cryptography to constrain the behavior of a quantum machine, it leaves a large gap between intuition and what is possible to prove. In particular,~\cite{oneproverrandomness} conjectured that in order to pass the protocol a prover must generate roughly $n$ bits of randomness per round, while the analysis could only guarantee a fraction of a bit.
 
 In this paper we show that a small change to the qubit certification test results in a protocol which  consists of two rounds of interaction, each communicating $\tilde{\Theta}(n)$ classical bits, and is such that any prover that succeeds with high enough constant probability in the protocol generates bits that are statistically close to a distribution with $\Theta(n)$ bits of min-entropy, even conditioned on the verifier's random questions. This conclusion holds as long as the prover is quantum polynomial-time bounded and the (post-quantum) Learning with Errors (QLWE) assumption holds for the duration of the protocol. We refer to Protocol~\ref{prot:amplifyrandomness} for a description of the protocol and Theorem~\ref{thm:43-randomness} for the formal guarantees. The small change that we make, which is described below, allows us to greatly simplify the proof of correctness while achieving a nearly optimal rate of randomness generation as a function of the total number of bits sent from the quantum device to the verifier, thereby also providing a way of certifying that the quantum device has roughly $n$ bits of quantum memory. We note that a quantum dimension test was recently obtained in~ \cite{fu2022computational} by performing parallel repetition of (a variant of) the protocol from~\cite{oneproverrandomness}. In contrast, our protocol is  more efficient (constant rounds, linear communication) and our analysis far more direct.\footnote{However,~\cite{fu2022computational}, see also~\cite{gheorghiu2022quantum}, proves a stronger rigidity result for their protocol, which extends the qubit test to certain $n$-qubit states.} 
 
 \medskip

 Beyond these quantitative improvements, our new protocol provides the first successful attempt to 
 move beyond NCTFs as a tool to help a classical verifier enforce structure in the quantum computer's Hilbert space. The main new ingredient in the improved protocol is a powerful notion from classical cryptography, called leakage resilience, which helps to more fully characterize the quantum computer.

Our point of departure is that whereas NTCFs are essentially $2$-to-$1$ functions, here  we consider cryptographic primitives that look like $k$-to-$1$ functions for exponentially large $k$, with the goal of certifying $\log k$ bits of randomness, rather than just $1$. More concretely, the $k$-to-$1$ function arises naturally from the use of leakage resilience techniques developed for the use of LWE in cryptography, and in particular the study of side channel attacks~\cite{goldwasser2010robustness}. The main idea to establish leakage resilience of LWE consists in replacing the uniformly random LWE matrix $A$  with a computationally indistinguishable lossy matrix $\tilde A = BC+E$ where $C$ has a large kernel, and $E$ is ``small''. Thus, up to the matrix $E$, which is needed to guarantee computational indistinguishability, $\tilde{A}$ is a close to a low-rank matrix---hence the terminology ``lossy''. 

To see how this would work let's start with a small change to the qubit certification protocol by switching from the use of an LWE matrix $A$ to a lossy matrix $\tilde A$. The reader familiar with the qubit certification protocol will see how at first glance this promises many more random bits. This is because, informally, in the protocol the quantum prover is tasked with preparing a uniform superposition $\frac{1}{\sqrt{2}}(\ket{x_0}+\ket{x_1})$ over two vectors $x_0$ and $x_1$ such that $Ax_0$ and $A(x_1-s)$ are sufficiently close, with $s$ a secret that is provided to the prover in encoded form. If the matrix $A$ is lossy then intuitively $\ket{x_0}$ would be replaced by $\frac{1}{\sqrt{k}}\sum_{i} |x_{0,i}\rangle$ where $x_{0,i}$ ranges over all preimages of $Ax_0$, and similarly for $\ket{x_1}$. Thus, by asking such a prover to measure its superposition in the computational basis one may hope to obtain approximately $ \log k$ bits of randomness as opposed to $1$. 

To keep the prover honest, the qubit certification protocol performs an equation test consisting of a measurement in the Hadamard basis --- the crux of protocol lies in the fact that passing this test implies that the prover must actually have a superposition over both $x_0$ and $x_1 = x_0 - s$. Unfortunately, the use of a lossy matrix and the appearance of multiple preimages completely breaks the structure of these tests, and it is unclear how to even test the prover's behavior for the case when $k$-to-$1$ functions are used. The actual protocol we analyze has the following structure: in the test round it uses a $2$-to-$1$ function, while in the randomness generation round it uses a $k$-to-$1$ function. To do this the verifier sends the prover an LWE sample $As +e$ in the test round and a lossy sample $\tilde{A}s +E$ in the generation round. From the viewpoint of the prover the two samples are indistinguishable, and one might hope that this keeps the prover honest in the generation round. Armed with this intuition, the challenge lies in showing that this simple test is sufficient to ensure that an untrusted prover is forced to output a large amount of min-entropy in the generation round.   

Unfortunately this idea faces an immediate difficulty: the tests used in the qubit certification protocol require the verifier to have knowledge of a secret trapdoor that allows inversion of the LWE matrix $A$. The use of this trapdoor rules out a direct kind of hybrid argument where the LWE matrix would be replaced by an indistinguishable lossy matrix while keeping the prover's success probability in the protocol unchanged. Our main idea to overcome this is to introduce a hypothetical \emph{quantum} verifier who is able to perform a meaningful test on the quantum prover \emph{without} making use of the trapdoor. This hypothetical verifier is used as an intermediate tool in the analysis, but the final protocol remains a classical verifier protocol. We explain these ideas in more detail in the next section, where we also precisely introduce the qubit  certification test to ground the discussion.

NTCFs have proved to be a powerful tool in constraining, characterizing and verifying the actions of untrusted quantum computers. This work can be viewed as a step towards leveraging more general cryptographic primitives, and suggests that leakage resilience may be a powerful new tool in this direction.

\paragraph{Acknowledgments}
T.V.\ is supported by AFOSR YIP award number FA9550-16-1-0495, a grant from the Simons Foundation (828076, TV), MURI Grant FA9550-18-1-0161, the NSF QLCI program through grant number OMA-2016245 and the IQIM, an NSF Physics Frontiers Center (NSF Grant PHY-1125565) with support of the Gordon and Betty Moore Foundation (GBMF-12500028). U.M.\ is supported by an NSF CAREER grant (2048204). U.V is supported by Vannevar Bush faculty fellowship N00014-17-1-3025, MURI Grant FA9550-18-1-0161, and DOE NQISRC Quantum Systems Accelerator grant FP00010905.

\section{Overview of results and techniques}

We begin with a reminder of the randomness protocol from \cite{oneproverrandomness}. For clarity we delay a discussion of parameters to Section~\ref{sec:prelim}. We let $q\geq 2$ be a prime and denote vectors and matrices over $\Z_q$ using bold font such as $\*u$, $\*A$.
The protocol from~\cite{oneproverrandomness}  (with some of the more technical details omitted for simplicity) consists in the sequential  repetition of the following elementary $2$-round protocol between a trusted classical \emph{verifier} and an untrusted quantum polynomial time device, or \emph{prover}:

\begin{protocol}{\textbf{Qubit Certification Test.}}\label{prot:randomness}
\begin{enumerate}
    \item The verifier samples $s\in\mZ_2^n$, $\*A\in\mZ_q^{m\times n}$ along with its trapdoor $t_A$ and an error vector $\*e\in\mZ_q^m$. The verifier computes $\*u = \*A s + \*e$.
    \item The verifier sends the sample $(\*A,\*u)$ to the prover.
    \item The prover reports an image $\*y\in\Z_q^m$ to the verifier.
    \item The verifier chooses to either run a test round or a generation round with equal probability $\frac{1}{2}$:
    \begin{enumerate}
        \item \textbf{Generation round.} The verifier runs a preimage test: he asks the prover for a bit $b$ and a preimage $\*x$ of $\*y$ and checks that $\lVert \*y - \*A\*x - b\cdot \*u\rVert$ is sufficiently small (at most $B_P\sqrt{m}$). 
        \item \textbf{Test round.} The verifier runs an equation test: the verifier uses $t_A$ to compute $\*x_0$ such that $\*y = \*A\*x_0 + \*e_0$ for some small $\*e_0$ and asks the prover for an equation, which consists of a bit $c$ and a string $d\in\{0,1\}^{n\log q}$. The verifier checks that $c = d\cdot (\*x_0\oplus (\*x_0 - s))$.\footnote{Recall from \cite{oneproverrandomness} that the bit $c$ is computed by converting $\*x_0, \*x_0 - s$ to their binary representations, and that there will also be a check on $d$ to ensure that the equation is not trivial.}
    \end{enumerate}
\end{enumerate}
\end{protocol}
To understand the rest of this overview, it is not necessary to understand all the details of the above protocol. It is important to understand what the honest prover's state looks like, which we now describe, as this will provide intuition for the changes we will make to the protocol. After reporting $y$ in Protocol \ref{prot:randomness}, an honest prover would hold the  state
\begin{equation}\label{eq:honestproverpreimage}
    \frac{1}{\sqrt{2}}\big(\ket{0}\ket{\*x_0} + \ket{1}\ket{\*x_0 - s}\big)\;.
\end{equation}
To pass the preimage test the prover would measure their state in the standard basis, and to pass the equation test they would measure the state in the Hadamard basis. Both tests would be passed with certainty.

In \cite{oneproverrandomness} it was shown that each generation round of the qubit certification test generates a fraction of a bit of randomness. The proof relies on a property called the \emph{adaptive hardcore bit property}, which (informally) states that if a computationally bounded prover can pass \textit{both} the equation test and the preimage test at the same time then the prover can break the Learning With Errors assumption by learning a bit of the secret $s$. The adaptive hardcore bit property implies that the prover must hold an essentially uniform superposition over both $x_0$ and $x_0 - s$, rather than a collapsed state. This is because if the superposition was not uniform then the prover could measure it in the computational basis without completely disturbing it, and then generate an equation. The prover would thus hold both an equation and a preimage at the same time, violating the adaptive hardcore bit property.

\medskip

The difficulty of learning even a bit of the secret $s$ is a commonly used notion in classical cryptography referred to as leakage resilience. To prove leakage resilience, the idea is to replace the matrix $\*A$ with a computationally indistinguishable \textit{lossy} matrix $\tilde{\*A} = \*B\*C + \*E$, where $\*C\in\mZ_q^{\ell \times n}$ ($n\approx \ell \log q)$). The computational indistinguishability follows from the security of a smaller instance of LWE (with secrets in $\mZ_q^{\ell}$ rather than $\mZ_q^n$). Now, the sample $\tilde{\*A}s + \*e$ hides $s$ quite well, as the matrix $\*C$ has a kernel of size $q^{n-\ell}$. As mentioned in the introduction there is one  particularly intuitive path to make use of leakage resilience  to collect a polynomial amount of randomness per round: we could replace the matrix $\*A$ in Protocol \ref{prot:randomness} with a lossy matrix $\tilde{\*A}$. As a result, for a given image $\*y$ reported by the prover there are now up to $q^{n-\ell}$ valid preimages for each $b$, rather than just 1; the hope is to leverage this structure to collect linearly many bits of randomness per generation round. 

The key difficulty lies in the equation test: the equation test is crucial for ensuring that the prover is actually generating randomness, as passing the equation test implies that the prover must actually have a superposition over both $\*x_0$ and $\*x_0 - s$. If the superposition over two preimages is replaced by a superposition over many, it is unclear how to perform a check in the Hadamard basis. Moreover, in the equation test in step 4(b) of Protocol \ref{prot:randomness}, the verifier relied on using a trapdoor to invert $\*y$, thereby finding $\*x_0$, which was used to compute $d\cdot (\*x_0\oplus (\*x_0 - s))$. Of course, the verifier can no longer rely on a trapdoor, as the image $\*y$ inherently has many valid preimages. Additionally, a useful analytical tool would be to compare the prover's behavior in the two different protocols arising from using either $\*A$ or $\tilde{\*A}$. If a trapdoor is required, we cannot use computational indistinguishability to carry out such an analysis.   

To handle both issues at once, consider a hypothetical variant of Protocol \ref{prot:randomness} in which the verifier is quantum rather than classical and is allowed access to the prover's state. Note that this hypothetical protocol will only be used as a tool in the analysis; our final protocol will still have a classical verifier. The advantage of using a quantum verifier is that the verifier's knowledge of the secret $s$, as well as access to the prover's state, can be used to compute $x_0$, thereby removing the need for a trapdoor. This can be done as long as the prover's state is in a superposition over preimages, which we can assume if the prover passes the preimage test perfectly. For intuition, observe how this can be done in the case of an honest prover. The prover's state (as in equation \ref{eq:honestproverpreimage}) is
\begin{equation}
    \frac{1}{\sqrt{2}}\big(\ket{0}\ket{\*x_0} + \ket{1}\ket{\*x_0 - s} \big)\;.\nonumber
\end{equation}
The verifier can copy the preimage ($\*x_0$ or $\*x_0 - s$) into an auxiliary register. He can then use his knowledge of $s$ to perform a controlled shift by $s$; if the first bit of the prover's state is 1, the verifier adds $s$ to the auxiliary register. As a result, the auxiliary register holds $\*x_0$, thereby removing the need for a trapdoor. Moreover, if the state was initially in a superposition over many preimages, rather than just 2, the verifier's actions will collapse the state to only two preimages, $\*x_0$ and $\*x_0 - s$, thereby making it possible to run the equation test. This protocol is referred to as the \emph{quantum verifier preimage extraction protocol}, and is described in detail in Protocol \ref{prot:extraction}.

The preceding idea is the core of our quantum verifier lossy randomness protocol (recall that it is a hypothetical protocol, which will only be used as an analytical tool). The protocol is essentially a duplicate of Protocol \ref{prot:randomness}, except with the matrix $\*A$ replaced by a lossy matrix $\tilde{\*A}$ and the trapdoor recovery of $\*x_0$ replaced by the quantum extraction protocol described above. 
\begin{protocol}{\textbf{Quantum Verifier Lossy Randomness Protocol}}\label{prot:lossyrandomness}
\begin{enumerate}
    \item The verifier samples $s\in\{0,1\}^n$, a lossy matrix $\tilde{\*A}\in\mZ_q^{m\times n}$ and an error vector $\*e\in\mZ_q^m$. The verifier computes $\*u = \tilde{\*A}s + \*e$.
    \item The verifier sends the sample $(\tilde{\*A},\*u)$ to the prover.
    \item The prover reports an image $\*y \in\mZ_q^m$ to the verifier.
     \item The verifier chooses to either run a test round or a generation round:
    \begin{enumerate}
        \item \textbf{Generation round: } The verifier runs a preimage test: he asks the prover for a bit $b$ and a corresponding preimage $\*x$ of $\*y$, and checks that $\lVert \*y - \tilde{\*A}\*x - b\cdot \*u\rVert$ is sufficiently small (at most $B_P\sqrt{m}$). 
        \item \textbf{Test round: } The verifier runs an equation test: the verifier uses the preimage extraction protocol (Protocol \ref{prot:extraction}) to compute $\*x_0$ from $\*y$, and then asks the prover for an equation, which consists of a bit $c$ and a string $d\in\{0,1\}^{n\log q}$. The verifier checks that $c = d\cdot (\*x_0\oplus (\*x_0 - s))$.
    \end{enumerate}

\end{enumerate}
\end{protocol}

In Theorem~\ref{thm:lossyhighentropy} we show that any prover that succeeds in the test round of Protocol \ref{prot:lossyrandomness} with high enough probability returns a pair $(b,\*x)$ in the generation round of the protocol that 
has min-entropy that scales as $n - \ell\log q$, even conditioned on data exchanged in the first round of the protocol (i.e.\ $\tilde{\*A}$, $\*u$ and $\*y$).\footnote{Precisely, we analyze the smooth min-entropy of the distribution: a prover with success $1-\eps$ in the test round will lead to a distribution that is within statistical distance $\poly(\eps)$ of a distribution with high conditional entropy.}
The proof is relatively simple, and consists of two key steps. The first step is observe that in order to pass the test round, the prover's collapsed state after the preimage extraction protocol must be in essentially a uniform superposition over the two preimages. To prove this, we use the fact that this statement is true in Protocol \ref{prot:randomness}; therefore, if it did not hold for Protocol \ref{prot:lossyrandomness}, it would be a means of distinguishing between a lossy matrix $\tilde{\*A}$ and a uniform matrix $\*A$. The second step relies on the fact that the secret $s$ is information theoretically hidden; even given the lossy sample $\*u$, there are exponentially many possible values for $s$. Therefore, prior to the verifier's preimage extraction protocol, the prover's state must be in an essentially uniform superposition over the exponentially many preimages of $\*y$, since the state must collapse to a balanced superposition over two preimages regardless of the verifier's choice of $s$. 

Our final protocol is a classical verifier lossy randomness protocol, which uses Protocol \ref{prot:lossyrandomness} as a tool in the analysis. To revert back to a classical verifier protocol, ideally we could just remove step 4(b) of Protocol \ref{prot:lossyrandomness}, as this is the only step that requires a quantum verifier. Of course, removing this step is problematic; without an equation test, we cannot guarantee that the prover ever holds a superposition over preimages (rather than just 1 preimage). Our solution is to replace step 4(b) of Protocol \ref{prot:lossyrandomness} with step 4(b) of Protocol \ref{prot:randomness}; in other words, our classical verifier lossy randomness protocol combines the test rounds of the original randomness protocol (Protocol \ref{prot:randomness}) with the generation rounds of our lossy randomness protocol (Protocol \ref{prot:lossyrandomness}):
\begin{protocol}{\textbf{Lossy Randomness Protocol}}\label{prot:amplifyrandomness}
The verifier randomly chooses whether to execute a generation round or a test round. In the case of a generation round, the verifier and prover proceed as in Protocol \ref{prot:lossyrandomness}, from the first step through the generation round (i.e. steps 1, 2, 3 and 4a). In the case of a test round, the verifier and prover proceed with Protocol \ref{prot:randomness} steps 1, 2, 3 and 4b.
\end{protocol}

Protocol~\ref{prot:amplifyrandomness} is restated in more detail in Section~\ref{sec:analysis-1}. 
The key in combining these protocols is that they are computationally indistinguishable: a prover passes the test (resp. generation) round of Protocol \ref{prot:randomness} if and only if they also pass the test (resp. generation) round of Protocol \ref{prot:lossyrandomness}. Otherwise, the prover could be used to distinguish between lossy and uniform matrices. This computational distinguishability argument relies on the fact that the probability of passing either protocol can be computed using a quantum algorithm with knowledge of the secret $s$; again, in Protocol \ref{prot:randomness}, the use of the trapdoor can be replaced with a quantum procedure with knowledge of $s$. Therefore, if the two probabilities differ non-negligibly, knowledge of $s$ can be used to distinguish between lossy and uniform matrices (using a quantum algorithm). These arguments are formalized in Section \ref{sec:analysis-1}.

\input{prelim.tex}

\input{analysis.tex}

\end{document}

%% file: def.tex
\newcommand{\bra}[1]{\left< #1\right|}
\newcommand{\ket}[1]{\left| #1\right>}

\newcommand{\Id}{\mathbbm{1}}
\newcommand{\mZ}{\mathbbm{Z}}

\newcommand{\ignore}[1]{}

\newtheorem{thm}{Theorem}[section]

\newtheorem{deff}[thm]{Definition}

\newtheorem{protocol}[thm]{Protocol}

\newtheorem{claim}[thm]{Claim}
\newtheorem{lem}[thm]{Lemma}

\def\hpic #1 #2 {\mbox{$\begin{array}[c]{l}
      \epsfig{file=#1,height=#2} \end{array}$}}

\def\vpic #1 #2 {\mbox{$\begin{array}[c]{l}
      \epsfig{file=#1,width=#2} \end{array}$}}

\DeclareMathOperator{\negl}{negl}

\newcommand{\qlwe}{\mathrm{QLWE}}
\newcommand{\lwe}{\mathrm{LWE}}
\newcommand{\poly}{\mathrm{poly}}

\newcommand{\vc}[1]{\mathbf{{#1}}}
\newcommand{\abs}[1]{\left\vert {#1} \right\vert}

\def\*#1{\mathbf{#1}}

\newcommand{\bbZ}{\mathbb{Z}}

\newcommand{\bbN}{\mathbb{N}}

\newcommand{\GenTrap}{\textsc{GenTrap}}
\newcommand{\Invert}{\textsc{Invert}}
\newcommand{\lossy}{\textsc{lossy}}

\newcommand{\Prov}{\mathds{P}}

\def\*#1{\mathbf{#1}}

\newcommand{\R}{\mathbb{R}}
\newcommand{\Z}{\mathbb{Z}}

\newcommand{\Es}[1]{\textsc{E}_{#1}}

%% file: prelim.tex
\section{Technical preliminaries}
\label{sec:prelim}

We start with some notation and definitions required for the formal analysis, which follows in Section~\ref{sec:analysis}. 
 
\subsection{Notation}

We use the shorthands PPT and QPT for probabilistic and quantum polynomial-time respectively. We write $\negl(\lambda)$ for any function $f:\mathbb{N}\to\R_+$ such that for any polynomial $p$, $\lim_{\lambda\to\infty} p(\lambda) f(\lambda)=0$. For $D$ a distribution on a finite set $\mathcal{X}$ we write $e\leftarrow D$ for a random element $e\in \mathcal{X}$ with distribution $D$. If $D$ is uniform we also write $e\leftarrow_U \mathcal{X}$. We use bold font such as $\*u, \*A$ to denote vectors or matrices whose entries are taken from $\mZ_q$.

\paragraph{Error distributions.}
For a positive real $B$ and a positive integer $q$, the 
truncated discrete Gaussian distribution over $\mZ_q$ with parameter $B$ is the distribution supported on $\{x\in\mZ_q:\,\|x\|\leq B\}$ with density
\begin{equation}\label{eq:d-bounded-def}
 D_{\mZ_q,B}(x) \,=\, \frac{e^{\frac{-\pi\lVert x\rVert^2}{B^2}}}{\sum\limits_{x\in\mZ_q,\, \|x\|\leq B}e^{\frac{-\pi\lVert x\rVert^2}{B^2}}} \;.
\end{equation}
More generally, for a positive integer $m$ the truncated discrete Gaussian distribution over $\mZ_q^m$ with parameter $B$ is the distribution supported on $\{x\in\mZ_q^m:\,\|x\|\leq B\sqrt{m}\}$ with density
\begin{equation}\label{eq:d-bounded-def-m}
\forall x = (x_1,\ldots,x_m) \in \mZ_q^m\;,\qquad D_{\mZ_q^m,B}(x) \,=\, D_{\mZ_q,B}(x_1)\cdots D_{\mZ_q,B}(x_m)\;.
\end{equation}

\paragraph{Min-entropy.} For random variables $X$ and $Y$ and $\eps\geq 0$ define the \emph{$\eps$-smooth conditional entropy of $X$ given $Y$} by
\[ H_\infty^\eps(X|Y) \,=\, \sup_{\tilde{X}\tilde{Y}:\, \|XY-\tilde{X}\tilde{Y}\|_{TV}\leq \eps} -\log\Big( \max_{y:\, \Pr(\tilde{Y}=y)>0}\; \max_x \;\Pr\big(\tilde{X}=x |\tilde{Y}=y\big)\Big)\;.\]
Here the supremum is taken over all sub-normalized pairs of random variables $(\tilde{X},\tilde{Y})$ with total variation distance at most $\eps$ from $(X,Y)$. 

\subsection{Parameters}

For convenience we collect here a description of the main parameters that are used in all our protocols. 
Let $\lambda$ be a security parameter. All other parameters are functions of $\lambda$. Let $q\geq 2$ be a prime integer. 
Let $\ell,n,m,w\geq 1$ be polynomially bounded functions of $\lambda$ and $B_L, B_V, B_P$ be positive integers such that the following conditions hold
\begin{enumerate}[label=({A.\arabic*})]
\item\label{a1}$n = \Omega(\ell \log q)$ and $m = \Omega(n\log q)$,
\item $w=n\lceil \log q\rceil$,
\item $B_P = \frac{q}{2C_T\sqrt{mn\log q}}$, for $C_T$ a universal constant,
\item $ 2\sqrt{n} \leq B_L < B_V < B_P$,
\item\label{a5} The ratios $\frac{B_P}{B_V}$ and $\frac{B_V}{B_L}$ are both super-polynomial  in $\lambda$.
\end{enumerate}
This setting of parameters is the same as the one in~\cite{oneproverrandomness}. The computational assumption that underlies our results is the same assumption under which the function family $\mathcal{F}_{\lwe}$ introduced in~\cite{oneproverrandomness} is an NTCF family, namely the Learning with Errors (LWE) assumptions with parameters $\ell,q$ and $B_L$. (In contrast to~\cite{oneproverrandomness}, since we do not consider quantum side information, we do not need to make the assumption with respect to quantum advice.)

\begin{deff}\label{def:lwe-ass}
For a security parameter $\lambda$, let $n,m,q\in \bbN$ be integer functions of $\lambda$. Let $\chi = \chi(\lambda)$ be a distribution over $\mZ$. The $\lwe_{n,m,q,\chi}$ problem is to distinguish between the distributions $(\*A, \*A\*s + \*e \pmod{q})$ and $(\*A, \*u)$, where $\*A\leftarrow_U \bbZ_q^{n \times m}$, $\*s\leftarrow_U \mZ_q^n$, $\vc{e}\leftarrow\chi^m$, and $\*u \leftarrow_U \mZ_q^m$. Often we consider the hardness of solving $\lwe$ for {any} function $m$ such that $m$ is at most a polynomial in $n \log q$. This problem is denoted $\lwe_{n,q,\chi}$. 

In this paper we make the assumption that no quantum polynomial-time procedure can solve the $\lwe_{n,q,\chi}$ problem with more than a negligible advantage in $\lambda$, when $\chi$ is the distribution $D_{\mZ_q^m,B}$ with $B\in \{B_L,B_V,B_P\}$. We refer to this assumption as ``the $\qlwe$ assumption,'' leaving the parameters $n$ and $q$ implicit. 
\end{deff}

The conditions~\ref{a1} to~\ref{a5} are such that two additional properties hold, which were already used in~\cite{oneproverrandomness}. Informally, the first property is the possibility to generate near-uniform LWE matrices $\*A$ together with a trapdoor allowing easy inversion; the second property is the existence of a computationally indistinguishable lossy mode. More formally, we make use of the following. 

\begin{thm}[Theorem 5.1 in~\cite{miccancio2012}]\label{thm:trapdoor}
Let $n,m\geq 1$ and $q\geq 2$ be such that $m = \Omega(n\log q)$. There is an efficient randomized algorithm $\GenTrap(1^n,1^m,q)$ that returns a matrix $\*A \in \mZ_q^{m\times n}$ and a trapdoor $t_{\*A}$ such that the distribution of $\*A$ is negligibly (in $n$) close to the uniform distribution. Moreover, there is an efficient algorithm $\Invert$ that, on input $\*A, t_{\*A}$ and $\*A\*s+\*e$ where $\|\*e\| \leq q/(C_T\sqrt{n\log q})$ and $C_T$ is a universal constant, returns $\*s$ and $\*e$ with overwhelming probability over $(\*A,t_{\*A})\leftarrow \GenTrap(1^n,1^m,q)$. \end{thm}

\begin{deff}\label{def:lossy}
Let $\chi = \chi(\lambda)$ be an efficiently sampleable distribution over $\mZ_q$. Define a lossy sampler $\tilde{\*A} \leftarrow \lossy(1^n,1^m,1^\ell,q,\chi)$ by  $\tilde{\*A} = \*B\*C +\*F$, where $\*B\leftarrow_U \mZ_q^{m\times \ell}$, $\*C\leftarrow_U \mZ_q^{\ell \times n}$, $\*F\leftarrow \chi^{m\times n}$. 
\end{deff}

\begin{thm}[Lemma 3.2 in~\cite{lwr}]\label{thm:lossy}
Under the $\qlwe$ assumption, the distribution of a random $\tilde{\*A} \leftarrow \lossy(1^n,1^m,1^\ell,q,\chi)$ is computationally indistinguishable from $\*A\leftarrow_U \mZ_q^{m\times n}$. 
\end{thm}

%% file: analysis.tex
\section{Analysis}
\label{sec:analysis}

Most of this section focuses on analyzing Protocol \ref{prot:lossyrandomness}.
 We begin by describing how to model a general prover; this characterization will be equivalent for all three protocols (Protocols \ref{prot:randomness}, \ref{prot:lossyrandomness}, \ref{prot:amplifyrandomness}). We then analyze Protocol \ref{prot:lossyrandomness}: we first restate the protocol, then show its completeness and finally soundness: for any computationally bounded prover that succeeds with large enough constant probability in the protocol, the distribution of $(b,x)$ obtained from a generation round contains a large amount of (smooth) conditional min-entropy (see Theorem~\ref{thm:lossyhighentropy}). Finally, in Section~\ref{sec:analysis-1} we explain how the analysis of Protocol \ref{prot:lossyrandomness} extends to Protocol \ref{prot:amplifyrandomness}.

\subsection{Prover Behavior}\label{sec:provercharacterization}

We begin by introducing a unified manner to refer to a prover's actions, whether the prover interacts with a verifier that executes Protocol~\ref{prot:randomness}, Protocol~\ref{prot:lossyrandomness} or Protocol~\ref{prot:amplifyrandomness}. All protocols have the following structure.

\begin{protocol}{\bf Protocol template}\label{prot:template}
\begin{enumerate}
    \item The verifier sends a pair $(\*A,\*u)\in \mZ_q^{m\times n} \times \mZ_q^m$ to the prover. (All parties are assumed to have agreed on integers $\lambda,m,n$ and $q$.)
    \item The prover reports a $\*y \in\mZ_q^m$ to the verifier.
     \item The verifier chooses a challenge $C\leftarrow_U\{G,T\}$ and sends $C$ to the prover.
    \begin{enumerate}
        \item (Case $C=G$, ``Generation round'':) The prover returns a bit $b$ and an $\*x\in\mZ_q^n$. 
        \item (Case $C=T$, ``Test round'':) The prover returns a bit $c$ and a string $d\in\{0,1\}^{n\lceil\log q\rceil}$. 
    \end{enumerate}
\end{enumerate}
\end{protocol}

Any prover in a protocol that has the same message structure as Protocol~\ref{prot:template} can be modeled as follows. The Hilbert space associated with the prover has four registers. The first register contains $m\lceil\log q\rceil$ qubits that will hold the provers' first answer $\*y$. The second register contains the first component of the prover's second answer and consists of a single qubit. The third register contains the second component of the prover's second answer and consists of $n\lceil\log q\rceil$ qubits. Finally, the fourth register contains the workspace of the prover (which is also used to store messages received from the verifier). 

Upon reception of the first message $(\*A,\*u)$ the prover may perform an arbitrary unitary $U'$ on all registers (including the private register, which is assumed to contain a copy of $(\*A,\*u)$) followed by a measurement of the first register in the standard basis to obtain its first response $\*y$. Upon reception of $C\in\{G,T\}$ the prover applies an arbitrary unitary $U'_C$ on their entire space, measures the second and third registers (collectively referred to as the \emph{preimage register}) and return the outcome obtained. 

In all protocols considered the honest behavior expected from the prover is the same, and is identical to the honest behavior in Protocol~\ref{prot:randomness} from~\cite{oneproverrandomness}. For clarity we do not reproduce all details here; informally, at the first step the prover is expected to prepare a uniform superposition over $(b,\*x)\in \{0,1\}\times \mZ_q^n$ and evaluate the function 
\[f_{\*A,\*u}:\;(b,\*x)\,\mapsto \,{\*A}\*x+b\, \*u + \*e\]
in superposition, where $\*e\leftarrow D_{\mZ_q^m,B_P}$. (The error vector $\*e$ also needs to be computed in superposition; we refer to~\cite[Section 5.3]{oneproverrandomness} for a detailed description of how this is performed.) The prover then measures the output register to obtain $\*y\in\mZ_q^m$. At the second step the prover is expected to measure the preimage register either in the computational basis (case $C=G$) or the Hadamard basis (case $C=T$) and report the outcome. Let $\hat{U}, \hat{U}_G=\Id$ and $\hat{U}_T=H^{\otimes(1+n\lceil \log q\rceil)}$ be the three unitaries associated with the honest prover behaviour. 

\begin{deff}
Suppose given a prover $\Prov$ for Protocol~\ref{prot:template} whose three unitaries are 
 $(U',U'_G,U'_T)$. Let $U_0 = \hat{U}^\dagger U'U'_G$ and $U=\hat{U}_G^\dagger (U'_G)^\dagger U'_T$. We say that the prover $\Prov$ is \emph{characterized by} $(U_0,U)$. 
\end{deff}

The definition is justified as follows. For any prover using unitaries $(U',U'_G,U'_T)$
 we may always assume that $U'_G$ and $U'_T$ both commute with a standard basis measurement of the prover's first register, because they are only applied after a measurement of it. Therefore, the prover is equivalently represented by the triple $(\hat{U}\hat{U}^\dagger U' U'_G, \Id, \hat{U}_G^\dagger (U'_G)^\dagger U'_T)$. It is thus without loss of generality that we restrict our attention to provers characterized by a pair $(U_0,U)$, where $U_0$ represents the prover's ``deviation'' from the honest prover in the first round, and $U$ its deviation in the second round, for the case $C=T$.


\subsection{Analysis of Quantum Verifier Lossy Randomness Protocol (Protocol \ref{prot:lossyrandomness})}\label{sec:lossyanalysis}

In this section we analyze Protocol \ref{prot:lossyrandomness}. For clarity we describe the protocol by filling in the verifier's actions in the template Protocol~\ref{prot:template}. First we need the following sub-protocol. This protocol describes an interaction between a quantum verifier and a quantum prover whose space has the same structure as considered in the other protocols, i.e.\ it consists of four registers, the second and third of which are collectively referred to as the ``preimage register''. In the protocol, the verifier is assumed to possess some $s\in\{0,1\}^n$.

\begin{protocol}{\textbf{Quantum Verifier Preimage Extraction Protocol}}\label{prot:extraction}
\begin{enumerate}
\item The prover sends their (quantum) preimage register to the verifier.
    \item  The verifier appends an auxiliary register of $n\lceil\log q\rceil$ qubits.
    \item The verifier coherently copies (in the standard basis) the contents of the second preimage register to the auxiliary register.
    \item  Conditioned on the first preimage register being in state $\ket{1}$ the verifier adds $s$ to the contents of the auxiliary register. (This is done coherently--if the first register is in state $\ket{0}$ then no action is performed.)
    \item The verifier measure the auxiliary register and stores the result. They return the preimage register to the prover. 
\end{enumerate}
\end{protocol}

We now restate the quantum verifier lossy randomness protocol. 

~\newline
\textbf{Protocol \ref{prot:lossyrandomness}: \textit{Quantum Verifier Lossy Randomness Protocol}}
\begin{enumerate}
    \item \textit{Let $\ell,n,m,q$ be functions of $\lambda$ satisfying the conditions~\ref{a1}--\ref{a5}. Let $\chi$ be the distribution $D_{\mZ_q^m,B_V}$. 
		The verifier samples $s\leftarrow_U \{0,1\}^n$, $\tilde{\*A}\leftarrow \lossy(1^n,1^m,1^\ell,q,D_{\mZ_q^m,B_L})$ and $\*e\leftarrow \chi$. The verifier computes $\*u = \tilde{\*A}s + \*e$. They send $(\tilde{\*A},\*u)$ to the prover.}
    \item \textit{The prover reports a $\*y \in\mZ_q^m$ to the verifier.}
     \item \textit{The verifier chooses to either run a generation round or a test round, with equal probability $1/2$ each.}
    \begin{enumerate}
        \item \textit{\textbf{Generation round:} The verifier sends $C=G$ to the prover and receives $(b,\*x)$ from them. They check that $\lVert \*y - \tilde{\*A}\*x - b\,\*u\rVert$ is at most $B_P\sqrt{m}$ and abort if not.  }
        \item \textit{\textbf{Test round:} The verifier uses the preimage extraction protocol (Protocol \ref{prot:extraction}) to compute $\*x_0$ from $\*y$. The verifier then sends $C=T$ to the prover and receives $(c,d)$ from them. The verifier checks that $d\neq 0$ and $c = d\cdot (J(\*x_0)\oplus J(\*x_0 - s))$, where $J$ is the binary representation function, which is applied coordinate-wise.}
    \end{enumerate}
\end{enumerate}

We note that in step 3(b) the prover is expected to return their preimage register exactly in the state that they are in after measurement of $\*y$, and as they would be directly measured in the standard basis in case $C=G$ (see Section~\ref{sec:provercharacterization}). One could call such provers ``semi-honest'' since at that step of the protocol they perform the ``honest'' action. Recall that Protocol \ref{prot:lossyrandomness} is a hypothetical protocol used for analysis purposes only, and in Section~\ref{sec:analysis-1} we will see that it is sufficient for us to analyze such ``semi-honest'' provers in the protocol. 

\subsubsection{Completeness}

We begin with proving completeness of Protocol \ref{prot:lossyrandomness}. To model the randomness generated in a run of the protocol we let $(\tilde{A},U,Y,C,B,X)$ be random variables associated with a transcript of the protocol in the obvious way. Specifically, 
\begin{itemize}
\item $(\tilde{A},U,Y)$ are distributed as $(\tilde{\*A},\*u,\*y)$ in the first round of Protocol~\ref{prot:lossyrandomness}. 
\item $C\in \{G,T\}$ is chosen uniformly at random. 
\item If $C=G$ (resp.\ $C=T$) then $(B,X)$ is distributed as the prover's answer in a generation round (resp.\ test round), conditioned on that type of round having been chosen by the verifier. \end{itemize}
We refer to $(\tilde{A},U,Y,C,B,X)$ as a \emph{random transcript} of the protocol.

\begin{lem}\label{lem:q-compl}
For any setting of parameters satisfying the assumptions~\ref{a1} to~\ref{a5}, there is a prover whose running time is polynomial in $\lambda$ such that the prover is accepted with probability negligibly close to $1$ in Protocol~\ref{prot:lossyrandomness}. Furthermore, there is a negligible function $\eps=\negl(\lambda)$ such that letting $(\tilde{A},U,Y,C,B,X)$ denote a random transcript of the protocol, it holds that $H^\eps_\infty(BX|\tilde{A}UY,C=G)\geq n-\ell\log q-O(1)$. 
\end{lem}

The following claim will be useful in showing the lemma (it is stronger than necessary, but is stated in this manner since it will also be used in the soundness proof): 

\begin{claim}\label{claim:size-0}
Let $\*u\in\mZ_q^n$. For $\tilde{\*A}\in \mZ_q^{m\times n}$ let $\mathcal{X}_{\tilde{\*A}} = \{\*x\in\mZ_q^n:\, \|\tilde{\*A}\*x-\*u\|\leq B_P\sqrt{m}\}$. Then
\begin{eqnarray}
    \Pr_{\tilde{\*A}\leftarrow\lossy(1^n,1^m,1^\ell,q,D_{\mZ_q^m,B_V})}\Big(\big|\mathcal{X}_{\tilde{\*A}}\cap\Z_2^n\big|\, \leq\, \frac{2^n-1}{2q^\ell} \Big) \,=\,\negl(\lambda)\;.
\end{eqnarray}
\end{claim}

\begin{proof}
Fix $\*u\in \mZ_q^n$. Recall from Definition~\ref{def:lossy} that $\tilde{\*A}$ can be written as $\tilde{\*A} = \*B\*C + \*E$, for $\*C\in \mZ_q^{\ell\times n}$. Since $\*C$ is selected uniformly at random, the probability that a fixed nonzero binary string maps to $\*u$ is exactly $\frac{1}{q^{\ell}}$. For $s\in \{0,1\}^n$ let $X_s$ be the indicator that $\*Cs=\*u$, so for $s\neq 0^n$, $\Es{}[X_s]=q^{-\ell}$. By linearity of expectation
\[ \Es{\tilde{\*A}} \big[ |\mathcal{X}_{\tilde{\*A}}\cap\Z_2^n|\big] \,=\, \sum_s \Es{\tilde{\*A}}\,[X_s] \,=\,\frac{2^n-1}{q^\ell}\;.\] 
Observe that $\{X_s\}_{s\in\{0,1\}^n}$ are pairwise independent, because two binary strings $s\neq s'$ are always linearly independent in $\mZ_q^n$. By applying Chebyshev's inequality we obtain
\begin{eqnarray}
    \Pr_{\tilde{\*A}}\Big(\big|\mathcal{X}_{\tilde{\*A}}\cap\Z_2^n\big| \,\leq\, \frac{2^n-1}{2q^\ell} \Big) \,\leq \,\frac{4q^\ell}{2^n-1}\;.
\end{eqnarray}
This probability is negligible as long as $n$ is sufficiently larger than $\ell \log q$, which is guaranteed by assumption~\ref{a1}.
\end{proof}

We now prove Lemma~\ref{lem:q-compl}.

\begin{proof}[Proof of Lemma~\ref{lem:q-compl}]
We use exactly the same prover as the one described in~\cite[Section 5.3]{oneproverrandomness}. For a generation round, the honest prover's state immediately after reporting $\*y$ in step 3 of Protocol \ref{prot:lossyrandomness} is within negligible trace distance of the state
\[|\psi\rangle \,=\,  \sum_{b\in\{0,1\}} \sum_{\*x_{b}\in\mathcal{X}_{\tilde{\*A},b,\*y}} \; \frac{1}{\sqrt{2|\mathcal{X}_{\tilde{\*A},b,\*y}|}} \ket{b}\ket{\*x_{b}}\;,\]
where $\mathcal{X}_{\tilde{\*A},b,\*y}$ denotes the set of valid preimages of $\*y$ under $\tilde{\*A}$, i.e.\ the set of all $\*x_{b,y}\in\Z_q^n$ such that $\| \*y - \tilde{\*A}\*x_{b,\*y} - b\, \*u\|\leq B_P\sqrt{m}$. This is because by definition the honest prover computes $\*y$ as $\*y=\tilde{\*A}\*x+b\, \*u + \*e$ for some vector $\*e$ such that $\|\*e\|\leq B_P\sqrt{m}$ (see Section~\ref{sec:provercharacterization}). 
Using Claim~\ref{claim:size-0}, with high probability over the choice of $\tilde{\*A}$ it holds that $|\mathcal{X}_{\tilde{\*A},b,\*y}|\geq 2^n/(2q^\ell)$. Once the preimage extraction protocol is applied in step 3(b) of Protocol \ref{prot:lossyrandomness}, the fact that the prover passes the test round with probability negligibly close to $1$ follows from the analysis of the honest strategy in~\cite{oneproverrandomness}. 
\end{proof}

\subsubsection{Soundness}

We will show that for any prover in Protocol \ref{prot:lossyrandomness}, as long as the prover succeeds with probability sufficiently high then 
the output distribution on pairs $(b,\*x)$ obtained in a generation round of Protocol \ref{prot:lossyrandomness} has high (smooth) min-entropy, even conditioned on all the information exchanged in the first round of the protocol, i.e.\ $(\tilde{\*A},\*u)$ and $\*y$. 

\begin{thm}\label{thm:lossyhighentropy}
Consider a quantum polynomial-time prover $\Prov$ in Protocol~\ref{prot:lossyrandomness} who passes the generation round with probability negligibly close to $1$ and the test round with probability $\frac{1}{2} + p_T$, for some $0< p_T\leq \frac{1}{2}$, and make the $\qlwe$ assumption. 
Let $(\tilde{A},U,Y,C,B,X)$ be a random transcript for an execution of the protocol. Then 
\begin{equation}\label{eq:min-bound}
 H_\infty^\varepsilon(BX|\tilde{A}UY, C=G) \,\geq\, n-\ell\log q -O(1)\;,
\end{equation}
where $\eps = 2(1-4p_T^2)^{1/4}+\negl(\lambda)$. 
\end{thm}

We note that the assumption that the prover succeeds in the generation round perfectly is essentially without loss of generality; see the proof of Theorem~\ref{thm:43-randomness} in Section~\ref{sec:analysis-1}. To prove the theorem we need the following refinement of Claim~\ref{claim:size-0}.

\begin{claim}\label{claim:size}
Let $\*u\in\mZ_q^n$. Then, with all but negligible probability over the choice of $\tilde{\*A}\in\mZ_q^{m\times n}$, 
\begin{eqnarray}
   \Pr_{\tilde{\*A}\leftarrow\lossy(1^n,1^m,1^\ell,q,D_{\mZ_q^m,B_V})} \Big( \max_{s\in\mZ_2^n}\Pr(s|\tilde{\*A}, \*u) > \frac{3q^\ell}{2^n-1} \Big) \,=\,\negl(\lambda)\;.
\end{eqnarray}
Here $\Pr(s|\tilde{\*A}, \*u) $ denotes the posterior probability of $s\in \{0,1\}^n$ having been sampled by the verifier at step 1 of Protocol~\ref{prot:lossyrandomness}, conditioned on $(\tilde{\*A},\*u)$ having been sent to the prover. 
\end{claim}

\begin{proof}
Recall from Definition~\ref{def:lossy} that a lossy matrix $\tilde{\*A}$ is equal to $\*B\*C + \*E$, for $\*B, \*C$ sampled uniformly from $\mZ_q^{m\times n}$ and $\mZ_q^{\ell\times n}$ respectively and $\*E\leftarrow D_{\mZ_q^{m\times n}, B_L}$. From Claim~\ref{claim:size-0} we already know that for any given $\*u\in \mZ_q^m$, with high probability over $\tilde{\*A}$, there are many possible  strings $s$ that could have led to $(\tilde{\*A},\*u)$. However, the posterior probability of each such $s$ is not identical, because different error vectors may have slightly different probabilities. To quantify this, define $D_{\tilde{\*A},\*u}$ to be the distribution on $\{0,1\}^n$ with probabilities 
\begin{eqnarray}\label{eq:def-DD}
    D_{\tilde{\*A},\*u}(s) \,=\, \frac{D_{\mZ_q^m,B_V}(\*u-\tilde{\*A}s)}{\sum\limits_{s'\in\{0,1\}^n}D_{\mZ_q^m,B_V}(\*u-\tilde{\*A}s')}\;.
\end{eqnarray}
Let $s_{0}\in\{0,1\}^n$ be a string which minimizes $\lVert \*u - \tilde{\*A}s_0 \rVert$ (and therefore maximizes $D_{\tilde{\*A},\*u}$), and let $\*e_{0} = \*u - \tilde{\*A}s_{0}$. Consider the set $S = \{s\in\{0,1\}^n | \*Cs = \*Cs_0\}$. For $s\in S$, $\tilde{\*A}s=\*B\*Cs+\*Es=\tilde{\*A}s_0+\*E(s-s_0)$, therefore by~\eqref{eq:def-DD},
\begin{eqnarray}
    D_{\tilde{\*A},\*u}(s) &\propto& D_{\mZ_q^m,B_V}(\*E(s_{0}-s) + \*e_{0})\;,
\end{eqnarray}
from which using~\eqref{eq:def-DD} it follows that
\begin{eqnarray}\label{eq:maxprobinitial}
    D_{\tilde{\*A},\*u}(s_{0}) 
    &\leq& \frac{D_{\mZ_q^m,B_V}(\*e_{0})}{\sum\limits_{s\in S}D_{\mZ_q^m,B_V}(\*E(s_{0}-s) + \*e_{0})}\;.
\end{eqnarray}
To lower bound the denominator we use the facts that for any binary $s$ and $\*E$ in the support of $D_{\mZ_q^{m\times n},B_L}$, $\lVert \*E(s_{0}-s)\rVert \leq 2B_L\sqrt{nm}$, and $\lVert \*e_{0}\rVert \leq B_V\sqrt{m}$. Using the definition of $D_{\mZ_q^m,B_V}$ we get that for all $s\in S$,
\begin{eqnarray*}
    D_{\mZ_q^m,B_V}(\*E(s_{0}-s) + \*e_{0}) &\geq& D_{\mZ_q^m,B_V}(\*e_{0})  e^{\frac{-\pi(2\lVert \*E(s_{0}-s)\rVert\lVert \*e_{0}\rVert )}{B_V^2}}e^{\frac{-\pi(\lVert \*E(s_{0}-s)\rVert^2 )}{B_V^2}}\\
    &\geq& D_{\mZ_q^m,B_V}(\*e_{0}) e^{\frac{-\pi(4B_L\sqrt{nm}B_V\sqrt{m} )}{B_V^2}}e^{\frac{-\pi(4B_L^2nm )}{B_V^2}}\\
    &\geq& D_{\mZ_q^m,B_V}(\*e_{0}) e^{\frac{-8\pi mn B_L }{B_V}}\;.
\end{eqnarray*}
Inserting this bound in \eqref{eq:maxprobinitial}, 
\begin{eqnarray*}
    D_{\tilde{\*A},\*u}(s_{0}) 
    &\leq& \frac{D_{\mZ_q^m,B_V}(\*e_{0})}{|S| D_{\mZ_q^m,B_V}(\*e_{0}) e^{\frac{-8\pi mn B_L }{B_V}}}\\
     &=& \frac{1}{|S| }e^{\frac{8\pi mn B_L }{B_V}}
\end{eqnarray*}
Using that by~\ref{a5} the ratio $B_L/B_V$ is negligible it follows that $D_{\tilde{\*A},\*u}(s_{0})\leq |S|^{-1}(1+\negl(\lambda))$. The claim follows by applying Claim~\ref{claim:size-0} to bound the size of $|S|$.  

\end{proof}

We now prove Theorem~\ref{thm:lossyhighentropy}.

\begin{proof}[Proof of Theorem~\ref{thm:lossyhighentropy}]
Fix $(\tilde{\*A},\*u,\*y)$ obtained in the first round of an execution of Protocol~\ref{prot:lossyrandomness} with the prover $\Prov$. Since the prover is assumed to pass the generation round with probability negligibly close to $1$, the state of the prover immediately prior to step 4 of the protocol is within negligible trace distance of a state of the form 
\begin{equation}\label{eq:lossyinitialstate}
    \sum_{b\in\{0,1\}} \sum_{\*x_{b}\in\mathcal{X}_{\tilde{\*A},b,\*y}} \; \sqrt{p_{b,\*x_{b}}}\ket{b}\ket{\*x_{b}}\ket{\psi_{b,\*x_{b}}}\;,
\end{equation}
where for $b\in\{0,1\}$, $\mathcal{X}_{\tilde{\*A},b,\*y}$ denotes the set of valid preimages of $\*y$, i.e.\ 
\[\mathcal{X}_{\tilde{\*A},b,\*y} \,=\, \big\{\*x_{b,\*y} \,|\, \lVert \*y - \tilde{A}\*x_{b,\*y} - b\, \*u\rVert\leq B_P\sqrt{m}\big\}\;\,\]
and the $\ket{\psi_{b,\*x_{b}}}$ are arbitrary normalized states. (We leave the dependence of $p_{b,\*x_b}$ and $\ket{\psi_{b,\*x_b}}$ on $\tilde{\*A}$, $\*u$ and $\*y$ implicit.) 
 We now show the following claim. 

\begin{claim}\label{claim:negl}
There exists a negligible function $\mu$ such that on average over the distribution of $(\tilde{\*A},\*u,\*y)$,
\begin{equation}\label{eq:firstaveragedifference}
\Es{(\tilde{\*A},\*u,\*y)}\Big[\sum_{\*x_{0}\in\mathcal{X}_{\tilde{\*A},0,\*y}}\big|p_{0,\*x_{0}}-p_{1,\*x_{0} - s}\big|\Big] \,\leq\, \sqrt{1-4p_T^2} + \mu\;,
\end{equation}
where $s$ is as chosen by the verifier in the first step of the protocol.
\end{claim}

\begin{proof}
Consider the state of the prover after the preimage extraction protocol that would be performed in step 3(b) of Protocol~\ref{prot:lossyrandomness}. Let $\*x_{0}$ denote the measurement result of the extraction. Based on~\eqref{eq:lossyinitialstate} and the definition of the preimage extraction protocol, Protocol~\ref{prot:extraction}, conditioned on $\*x_0$ being obtained the state of the prover has  collapsed to
\begin{equation}\label{eq:lossycollapsedstate}
    \frac{1}{\sqrt{p_{0,\*x_{0}} + p_{1,\*x_{1}}}}\sum_{b\in\{0,1\}} \sqrt{p_{b,\*x_{b}}}\ket{b}\ket{\*x_{b}}\ket{\psi_{b,\*x_{b}}}\;,
\end{equation}
where $\*x_{1} = \*x_{0} - s$.
By assumption the probability (over all previous outcomes obtained in the protocol) of obtaining a correct equation using the prover's measurement is $\frac{1}{2} + p_T$. Suppose now that $(b,\*x_b)$ is measured prior to performing the prover's measurement. We will use the following claim.

\begin{claim}\label{claim:overlap}
Let $\ket{\psi} = \alpha\ket{0} + \beta\ket{1}$ be normalized and $H$ Hermitian such that $\|H\|\leq 1$. Then 
\[ |\alpha|^2 \bra{0}H\ket{0} + |\beta|^2 \bra{1}H\ket{1} \,\geq\, \bra{\psi} H \ket{\psi} - 2|\alpha||\beta|\;.\]
\end{claim}

\begin{proof}
By direct calculation,
\begin{align*}
\bra{\psi} H \ket{\psi} &= |\alpha|^2 \bra{0}H\ket{0} + |\beta|^2 \bra{1}H\ket{1} + 2\Re\big( \alpha\overline{\beta} \bra{0}H\ket{1}\big)\\
&\leq |\alpha|^2 \bra{0}H\ket{0} + |\beta|^2 \bra{1}H\ket{1} + 2|\alpha||\beta|\;.
\end{align*}
\end{proof}

Applying Claim~\ref{claim:overlap}, using $H$ the observable that combines the prover's equation measurement with the check that the equation is correct, it follows  that the probability of obtaining a correct equation when $(b,\*x_b)$ is measured {prior} to performing the provers' equation measurement (summed over all possible $(b,\*x_b$) is at least
\begin{equation} 
\frac{1}{2} + p_T - \Es{(\tilde{\*A},\*u,\*y)}\Big[\sum_{\*x_{0}\in\mathcal{X}_{\tilde{\*A},0,\*y}} \sqrt{p_{0,\*x_{0}}p_{1,\*x_{0}-s}}\Big]\;.
\end{equation}
It must be the case that obtaining a correct equation following a measurement of $(b,\*x_b)$ occurs with probability at most $\frac{1}{2} + \mu'$, for a negligible function $\mu'$. This is because if the probability is non-negligibly larger than $\frac{1}{2}$, prover $\Prov$ and the verifier can be turned into an efficient adversary for the adaptive hardcore bit property shown in~\cite[Section 4.4]{oneproverrandomness}. 
Therefore,
\begin{equation}\label{eq:p-1}
  \Es{(\tilde{\*A},\*u,\*y)}\Big[\sum_{x_{0}\in\mathcal{X}_{\tilde{\*A},0,\*y}} \sqrt{p_{0,\*x_{0}}p_{1,\*x_{0}-s}}\Big] \,\geq\,\frac{1}{2} + p_T - \Big(\frac{1}{2} + \mu'\Big)\;.
\end{equation}
Using the Cauchy-Schwarz inequality,
\begin{eqnarray}
  \sum_{\*x_{0}\in\mathcal{X}_{\tilde{\*A},0,\*y}} \big| p_{0,\*x_{0}} - p_{1,\*x_{0}-s}\big| &=&  \sum_{\*x_{0}\in\mathcal{X}_{\tilde{\*A},0,\*y}} \big|(\sqrt{p_{0,\*x_{0}}} - 
    \sqrt{p_{1,\*x_{0}-s}})(\sqrt{p_{0,\*x_{0}}} +
    \sqrt{p_{1,\*x_{0}-s}})\big|\notag\\
    &\leq&\Big(\sum_{\*x_{0}\in\mathcal{X}_{\tilde{\*A},0,\*y}} \big(\sqrt{p_{0,\*x_{0}}} - 
    \sqrt{p_{1,\*x_{0}-s}}\big)^2\Big)^{\frac{1}{2}}\Big(\sum_{\*x_{0}\in\mathcal{X}_{\tilde{\*A},0,\*y}} \big(\sqrt{p_{0,\*x_{0}}} + 
    \sqrt{p_{1,\*x_{0}-s}}\big)^2\Big)^{\frac{1}{2}}\notag\\
&\leq&  \Big(1 - 4  \Big(\sum_{\*x_{0}\in\mathcal{X}_{\tilde{\*A},0,\*y}} \sqrt{p_{0,\*x_{0}}p_{1,\*x_{0}-s}} \Big)^2 \Big)^{1/2}\label{eq:cauchyschwarz}
\end{eqnarray}
Combining \eqref{eq:p-1}, \eqref{eq:cauchyschwarz}, and using that $x\mapsto \sqrt{1-4x^2}$ is concave on $[0,1/2]$ proves the claim.
\end{proof}

Next we take advantage of the fact that the verifier's sample $(\tilde{\*A},\*u)$ is lossy to extend \eqref{eq:firstaveragedifference} to all $s'$ which could have resulted in the sample; in other words, all binary $s'$ for which there exists $\*e'$ (with norm at most $B_V\sqrt{m}$) such that $\tilde{\*A}s' + \*e' = \*u$. For any such $s'$, with respect to the first round of the protocol the verifier's secret may as well have been $s'$ rather than $s$. Since the $p_{0,\*x_0}$  defined from the prover's state at the end of the first round depend only on $(\tilde{\*A},\*u,\*y)$ but not directly on $s$ it follows from Claim~\ref{claim:negl} that there exists a negligible function $\mu$ such that on average over $(\tilde{\*A},\*u,\*y)$,
\begin{align}
   \mathrm{E}_{(\tilde{\*A},\*u,\*y)}\Big[ \sum_{\*x_{0}\in\mathcal{X}_{\tilde{\*A},\*0,\*y}} \abs{p_{0,\*x_{0}}-\mathrm{E}_{s'}[p_{1,\*x_{0} - s'}]} \Big]
	&\leq 
   \mathrm{E}_{(\tilde{\*A},\*u,\*y)}\Big[\mathrm{E}_{s'}\Big[\sum_{\*x_{0,\*y}\in\mathcal{X}_{\tilde{\*A},0,\*y}}\abs{p_{0,\*x_{0,\*y}}-p_{1,\*x_{0,\*y} - s'}}\Big]\Big]\notag\\
	&= 
 \mathrm{E}_{s'} \Big[ \mathrm{E}_{(\tilde{\*A},\*u,\*y)}\Big[\sum_{\*x_{0,\*y}\in\mathcal{X}_{\tilde{\*A},0,\*y}}\abs{p_{0,\*x_{0,\*y}}-p_{1,\*x_{0,\*y} - s'}}\Big]\Big]\notag\\
&\leq \sqrt{1-4p_T^2} + \mu\;,\label{eq:p-2}
\end{align}
where here in the first line the expectation on $s'$ is taken under the marginal distribution, conditioned on $(\tilde{\*A},\*u)$, and in the second line it is uniform over $\{0,1\}^n$. The first inequality is Jensen's inequality.  
Using Claim~\ref{claim:size} it follows that except with negligible probability over the choice of $\tilde{\*A}$ and $\*u$,
\begin{equation}\label{eq:p-3}
 \mathrm{E}_{s'}\big[p_{1,\*x_{0,\*y} - s'}\big] \,\leq \,\frac{3q^\ell}{2^n-1}\;.
\end{equation}
Combining~\eqref{eq:p-2} and~\eqref{eq:p-3}, by Markov's inequality with probability at least  $(1-4p_T^2)^{1/4}$ over the choice of $(\tilde{\*A},\*u,\*y)$ it holds that the distribution of $(b,\*x_b)$ obtained in the generation round, conditioned on $b=0$, is within statistical distance $(1-4p_T^2)^{1/4} + \mu'$ of a distribution with entropy at least $n-\ell\log q-O(1)$. The same reasoning applies to the distribution over preimages conditioned on $b=1$. This concludes the proof. \end{proof}

\subsection{Analysis of Lossy Randomness Protocol (Protocol \ref{prot:amplifyrandomness})}\label{sec:analysis-1}

In this section we describe how Theorem~\ref{thm:lossyhighentropy}, which applies to Protocol~\ref{prot:lossyrandomness}, 
can be extended to guarantee randomness generation for Protocol \ref{prot:amplifyrandomness}. For convenience we first restate Protocol~\ref{prot:amplifyrandomness} in the same format as Protocol~\ref{prot:lossyrandomness}, by filling in the verifier's actions in the template Protocol~\ref{prot:template}.

~\newline
\textbf{Protocol \ref{prot:amplifyrandomness}: \textit{Lossy Randomness Protocol}}
\begin{enumerate}
    \item \textit{Let $\ell,n,m,q$ be functions of $\lambda$ satisfying the conditions~\ref{a1}--\ref{a5}. Let $\chi$ be the distribution $D_{\mZ_q^m,B_V}$. The verifier chooses $C\in\{G,T\}$ uniformly at random. }
		\begin{enumerate}
		\item  \textit{If $C=G$ the verifier samples $s\leftarrow_U \{0,1\}^n$, $\tilde{\*A}\leftarrow \lossy(1^n,1^m,1^\ell,q,D_{\mZ_q^m,B_L})$ and $\*e\leftarrow \chi$. The verifier computes $\*u = \tilde{\*A}s + \*e$. They send $(\tilde{\*A},\*u)$ to the prover.}
		\item  \textit{If $C=T$ the verifier samples  $s\leftarrow_U \{0,1\}^n$, $({\*A},t_{\*A})\leftarrow \GenTrap(1^n,1^m,q)$ and $\*e\leftarrow \chi$. The verifier computes $\*u = {\*A}s + \*e$. They send $({\*A},\*u)$ to the prover.}
		\end{enumerate}
	  \item \textit{The prover reports a $\*y \in\mZ_q^m$ to the verifier.}
     \item \textit{The verifier runs either a generation round (in case $C=G$) or a test round (in case $C=T$).}
    \begin{enumerate}
        \item \textit{\textbf{Generation round:} The verifier sends $C$ to the prover and receives $(b,\*x)$ from them. They check that $\lVert \*y - \tilde{\*A}\*x - b\, \*u\rVert$ is at most $B_P\sqrt{m}$ and abort if not.  }
        \item \textit{\textbf{Test round:} The verifier sends $C$ to the prover and receives $(c,d)$ from them. They use the trapdoor $t_{\*A}$ to compute $\*x_0\leftarrow \Invert(\*A,t_{\*A},\*y)$, where $\Invert$ is the algorithm from Theorem~\ref{thm:trapdoor}. They check that $d\neq 0$ and $c = d\cdot (J(\*x_0)\oplus J(\*x_0 - s))$, where $J$ is the binary representation function, and abort if either condition is not satisfied.}
    \end{enumerate}
\end{enumerate}

We start with two preliminary claims that allow us to relate the success probability of a prover in the generation or test rounds of any of the three protocols \ref{prot:randomness}, \ref{prot:lossyrandomness}, and \ref{prot:amplifyrandomness}).

\begin{claim}\label{cl:perfectgenerationrounds}
Consider a prover who passes the generation round of either one of the three protocols (Protocols \ref{prot:randomness}, \ref{prot:lossyrandomness}, \ref{prot:amplifyrandomness}) with probability $ p_G$. Such a prover passes the generation rounds of all three protocols with probability negligibly close to $p_G$.  
\end{claim}

\begin{proof}
All provers must have the same probability of passing the generation rounds of protocols \ref{prot:lossyrandomness} and \ref{prot:amplifyrandomness}, since they are identical. If this probability differed from the probability of passing the generation round of protocol \ref{prot:randomness}, the prover could be used as an efficient means of distinguishing between $\tilde{\*A}\leftarrow \lossy(1^n,1^m,1^\ell,q,D_{\mZ_q^m,B_L})$ and ${\*A}\leftarrow \GenTrap(1^n,1^m,q)$, since given $\tilde{\*A}$ or $\*A$ estimating the prover's acceptance probability in a generation round can be done efficiently.  However, by Theorem~\ref{thm:trapdoor} the distribution of $\*A$ is statistically close to uniform, whereas by Theorem~\ref{thm:lossy} the distribution of $\tilde{\*A}$ is computationally indistinguishable from uniform. 
\end{proof}

\begin{claim}\label{cl:equivalenceoflossyandLWErandomness}
Consider a prover who passes the generation rounds of any one of the three protocols (Protocols \ref{prot:randomness}, \ref{prot:lossyrandomness}, \ref{prot:amplifyrandomness}) with probability negligibly close to $1$. For such a prover, the probabilities of passing the test rounds of each of the three protocols is negligibly close.  
\end{claim}

\begin{proof}
From Claim \ref{cl:perfectgenerationrounds}, such a prover must pass the generation rounds of all three protocols with all but negligible probability. 
If we can show that computing the probability of passing the test rounds of each of the three protocols does not require a trapdoor (and is therefore efficient), then the proof is complete, as otherwise, similarly to the proof of Claim~\ref{cl:perfectgenerationrounds} differing probabilities in the test rounds would serve as a means of efficiently distinguishing between lossy $\tilde{\*A}$ and uniform $\*A$. 

Computing the probability of passing the test round of Protocol \ref{prot:lossyrandomness} does not require the trapdoor, due to the quantum verifier extraction protocol (Protocol \ref{prot:extraction}). We can make the same claim for Protocols~\ref{prot:randomness} and~\ref{prot:amplifyrandomness} by replacing the trapdoor usage in step 3(b) with the quantum verifier extraction protocol. This change will not affect the prover's probability of succeeding in the test round because for provers who pass the generation round with probability negligibly close to $1$ this replacement does not change the probability of acceptance in the test round except by a negligible amount.
\end{proof}

We can now conclude by stating the completeness and soundness properties of Protocol~\ref{prot:amplifyrandomness}. First we note that the same completeness statement as for Protocol~\ref{prot:lossyrandomness} in Lemma~\ref{lem:q-compl} remains valid for Protocol~\ref{prot:amplifyrandomness}, with the same parameters; moreover, the behavior of the honest prover is exactly the same. For soundness, using the preceding claims we can show the following by reduction to Theorem~\ref{thm:lossyhighentropy}.

\begin{thm}\label{thm:43-randomness}
Consider a quantum polynomial-time prover $\Prov$ in Protocol~\ref{prot:amplifyrandomness} who passes the generation round with probability $p_G$ and the test round with probability $\frac{1}{2} + p_T$, for some $0< p_G\leq 1$ and $0\leq p_T\leq \frac{1}{2}$ such that $\delta = \frac{1}{2}-p_T+\sqrt{1-p_G}\leq \frac{1}{2}$, and make the $\qlwe$ assumption. 
Let $(\tilde{A},U,Y,C,B,X)$ be a random transcript for an execution of the protocol. Then 
\[ H_\infty^\varepsilon(BX|\tilde{A}UY, C=G) \,\geq\, n-\ell\log q -O(1)\;,\]
where $\eps = 5\delta^{1/4}+\negl(\lambda)$. 
\end{thm}

\begin{proof}
First we observe that given a prover $\Prov$ that passes the generation (resp.\ test) round in Protocol~\ref{prot:amplifyrandomness} with probability $p_G$ (resp.\ $p_T$), it is straightforward to construct a prover $\Prov'$ that succeeds in the generation round with probability negligibly close to $1$, and the test round with probability $\frac{1}{2}+p'_T$ for some  $p'_T \geq p_T-\sqrt{1-p_G}-\negl(\lambda)$. This follows from~\cite[Lemma 7.2]{oneprover} and only uses that  checking acceptance in the generation round is efficient. Moreover, the entire quantum state of $\Prov$ and $\Prov'$ after having reported $\*y$ in the first round are within trace distance $\sqrt{1-p_G}+\negl(\lambda)$, and the provers' actions are identical from there on. 

Applying Claim \ref{cl:equivalenceoflossyandLWErandomness}, we deduce that $\Prov'$ passes the test round of Protocol~\ref{prot:lossyrandomness} with probability negligibly close to $\frac{1}{2}+p'_T$. This allows us to apply Theorem~\ref{thm:lossyhighentropy} and conclude that, when $\Prov'$ is used in an execution of Protocol~\ref{prot:lossyrandomness}, the transcript $(\tilde{A},U,Y,C,B,X)$ of the protocol satisfies~\eqref{eq:min-bound}. Using that the distribution of outcomes generated by $\Prov'$ and $\Prov$ are within statistical distance $\sqrt{1-p_G}+\negl(\lambda)$, the theorem follows. \end{proof}

\bibliographystyle{alpha}
\bibliography{qpip}